\documentclass{article}

    \PassOptionsToPackage{numbers, compress}{natbib}

 \usepackage[dblblindworkshop, final]{neurips_2025}

\workshoptitle{Generative AI in Finance}

\usepackage[utf8]{inputenc} 
\usepackage[T1]{fontenc}    
\usepackage{hyperref}       
\usepackage{url}            
\usepackage{booktabs}       
\usepackage{amsfonts}       
\usepackage{nicefrac}       
\usepackage{microtype}      
\usepackage{xcolor}         
\usepackage{amsmath}
\usepackage{amsthm}
\newtheorem{theorem}{Theorem}
\newtheorem*{theoremNoNum}{Theorem}
\newtheorem{lemma}{Lemma}
\usepackage{algorithm}
\usepackage{algpseudocode}
\title{Shift-Aware Gaussian-Supremum Validation for Wasserstein-DRO CVaR Portfolios}

\author{%
  Derek Long\\
  Department of Industrial Engineering and Operations Research\\
  Columbia University\\
  New York, NY 10027\\
  \texttt{d.long@columbia.edu} \\
}

\newcommand{\R}{\mathbb{R}}
\newcommand{\E}{\mathbb{E}}

\begin{document}

\maketitle

\begin{abstract}
We study portfolio selection with a Conditional Value-at-Risk (CVaR) constraint under distribution shift and serial dependence. While Wasserstein distributionally robust optimization (DRO) offers tractable protection via an ambiguity ball around empirical data, choosing the ball radius is delicate: large radii are conservative, small radii risk violation under regime change. We propose a shift-aware Gaussian-supremum (GS) validation framework for Wasserstein-DRO CVaR portfolios, building on the work by Lam and Qian (2019). Phase~I of the framework generates a candidate path by solving the exact reformulation of the robust CVaR constraint over a grid of Wasserstein radii. Phase~II of the framework learns a target deployment law $Q$ by density-ratio reweighting of a time-ordered validation fold, computes weighted CVaR estimates, and calibrates a simultaneous upper confidence band via a block multiplier bootstrap to account for dependence. We select the least conservative feasible portfolio (or abstain if the effective sample size collapses). Theoretically, we extend the normalized GS validator to non-i.i.d.\ financial data: under weak dependence and regularity of the weighted scores, any portfolio passing our validator satisfies the CVaR limit under $Q$ with probability at least $1-\beta$; the Wasserstein term contributes a deterministic margin $(\delta/\alpha)\|x\|_*$. Empirical results indicate improved return--risk trade-offs versus the naive baseline.
\end{abstract}

\section{Introduction}
Portfolio optimization under uncertainty requires robust risk management, especially for extreme losses. A common risk metric is the Conditional Value-at-Risk (CVaR), which measures the average loss in the worst $\alpha$ fraction of cases. CVaR constraints are used to cap tail risk and prevent systemic crises by ensuring portfolio losses remain below a threshold with high confidence. In practice, data-driven portfolio selection with a CVaR constraint faces significant challenges: the true distribution of returns is unknown and may shift over time (regime shifts), and asset returns are often dependent (serially correlated and cross-correlated), violating i.i.d. assumptions. Classical robust optimization or distributionally robust optimization (DRO) approaches account for some distributional uncertainty by considering worst-case distributions in an ambiguity set (e.g. a Wasserstein ball around the empirical distribution). However, these methods can be sensitive to unanticipated distribution shifts: if market conditions change (e.g. a volatility regime change), a previously feasible portfolio may violate the CVaR limit. Moreover, dependent data reduces effective sample size and can lead to mis-estimated risk. This motivates the need for a shift-aware validation procedure that guarantees CVaR constraints under distribution shifts and dependence.


We focus on the CVaR-constrained portfolio problem with $d$ assets. Let $x\in\R^d$ be allocations, $\xi\in\R^d$ asset returns, $F$ the (unknown) data-generating law, and $c=-\E_F[\xi]\in\R^d$ (expected losses). The classical formulation is
\begin{equation}\label{cvar}
\begin{aligned}
& \underset{x}{\mathrm{min}}
& & c^\top x \\
& \mathrm{s.t.}
& & \inf_{t\in\R}\Big\{t+\tfrac{1}{\alpha}\E_F[\max\{-\xi^\top x - t,\,0\}]\Big\}\le \gamma,\\
&&& \sum_{l=1}^d x_l=1,\qquad x_l\ge 0~\forall l\in\{1,\ldots,d\}.
\end{aligned}
\end{equation}
To guard against distributional uncertainty, we adopt a Wasserstein DRO version centered at an empirical law $\hat F$ with radius $\delta$, so that the robust CVaR constraint becomes
\begin{equation}\label{robust cvar}
\sup_{Q\in\mathcal B(\hat F,\delta)} \inf_{t\in\R}\Big\{t+\tfrac{1}{\alpha}\E_Q[\max\{-\xi^\top x - t,\,0\}]\Big\}\le \gamma.\\
\end{equation}


For $W_1$ ambiguity balls, the robust CVaR constraint admits an exact, tractable reformulation (cf.\ Cor.~5.1 in~\cite{esfahani2018wass}):
\begin{equation}\label{cvar:reformulation}
Z=\left\{x\in\R^d:\begin{array}{l}
\delta v+ \frac{1}{n}\sum_{i=1}^n z_i\leq \alpha r,\\[0.2em]
r\leq z_i+\gamma+\xi_i^\top x\ \ \text{for }i=1,\ldots,n,\\[0.2em]
\|x\|_*\leq v,\\
v\geq 0,\ \ r\geq 0,\ \ z_i\geq 0\ \ \text{for }i=1,\ldots,n,
\end{array}\right\},
\end{equation}
where $\|\cdot\|_*$ is the dual norm tied to the ground metric. While \eqref{cvar:reformulation} is computationally convenient, selecting $\delta$ is delicate: too small jeopardizes feasibility under shift; too large kills performance. Our contribution is a data-driven selection via a \emph{shift-aware} GS validator that yields feasibility guarantees under serial dependence.

\section{Related Work and Contributions}\label{related-work-contributions}
Wasserstein DRO delivers tractable, statistically meaningful robustness~\citep{esfahani2018wass, kuhn2019tutorial}, with applications to finance and CVaR portfolios~\citep{blanchet2022drmv, esfahani2018wass, pun2023regimeswitching}. However, calibrating the radius can be challenging and often leads to overly conservative solutions. To avoid these overly conservative solutions, Lam and Qian~\cite{lamqian2019gs} develop a normalized Gaussian-supremum validator (among other validators) for data-driven optimization with uncertain constraints. Instead of fixing a very conservative reformulation (large safety margins), their approach generates a solution path of candidate portfolios indexed by a conservativeness parameter (e.g., different values of $\rho$ in the DRO formulation). They then evaluate these candidates on a holdout dataset to find the least conservative solution that still meets the constraint with high confidence ($1-\beta$). The GS validator computes a simultaneous confidence bound across the entire solution path by estimating the quantile of the supremum of Gaussian deviations.

In machine learning, covariate shift refers to the scenario where the input distribution changes between training and testing phases, while the conditional relationship (e.g., returns given factors) remains stable. A rich body of work addresses this via importance weighting techniques: direct ratio estimation (KLIEP, uLSIF)~\citep{sugiyama2012dre,kanamori2009ulsif}, kernel mean matching (KMM)~\citep{huang2007kmm}, and classifier-based ratios~\citep{bickel2009covshift}; these ideas enable targeting a \emph{shifted} law in validation. Furthermore, standard statistical inference assumes independent data, but financial time series exhibit autocorrelation, volatility clustering, and other dependencies. To overcome this, time-series inference under dependence uses block bootstraps~\citep{kunsch1989jackknife,politis1994stationary} and multiplier/wild bootstraps for dependent data~\citep{shao2010dependent,bucher2016dependent}. Our validation builds on the \emph{Gaussian-supremum validator} of Lam \& Qian~\citep{lamqian2019gs}; we extend it to non-i.i.d.\ financial data with shift-aware weighting and BMB calibration.

\textbf{Contributions.} \emph{(i)} A \textbf{shift-aware GS validator}: we reweight validation data by estimated density ratios to evaluate CVaR under a learned center $Q$, then perform Wasserstein DRO around $Q$ with a smaller radius. \emph{(ii)} A \textbf{dependence-robust calibration}: we replace i.i.d.\ Gaussian calibration by a block multiplier bootstrap for the GS band, yielding correct joint coverage under serial dependence. \emph{(iii)} \textbf{Finite-sample feasibility}: we prove that any selected portfolio meets the CVaR limit under $Q$ with probability $\ge 1-\beta$; proofs follow GS arguments with extensions for weighting and dependence.

\section{Methodology and Main Results}\label{sec:methodology}
We adopt a two-phase procedure. Phase~I generates a candidate path $\{x^*(\delta_j)\}_{j=1}^p$ by solving \eqref{cvar:reformulation} over a grid of radii $\delta_1<\cdots<\delta_p$. Phase~II validates feasibility on a holdout, time-ordered sample $\{\xi_t\}_{t=1}^{n_2}$ potentially from a shifted regime. The full algorithm is detailed in Appendix \hyperref[appendix-A]{A}.

\textbf{Shift-aware center via density ratios.} Let $P$ denote the holdout law and $Q$ the target law for deployment (current regime). We estimate $w(\xi)=q(\xi)/p(\xi)$ using a probabilistic classifier that distinguishes recent (proxy for $Q$) vs.\ holdout (proxy for $P$) samples, convert scores to odds, \emph{clip}, and normalize: $\tilde w_t=w_t/\sum_{i}w_i$. For any candidate $x$, we estimate the CVaR functional’s expectation under $Q$ by the weighted mean:
\[
\widehat\Psi_Q(x)\;=\;\min_{t\in\R}\left\{t+\tfrac{1}{\alpha}\sum_{i=1}^{n_2}\tilde w_i\,\max\{-\xi_i^\top x - t,\,0\}\right\}.
\]
Effective sample size is $n_{\text{eff}}=1/\sum_i \tilde w_i^2$; we abstain if $n_{\text{eff}}$ collapses.

\textbf{Dependence-robust GS calibration.} For each candidate $x_j=x^*(\delta_j)$ with empirical minimizer $\hat t_j$, define scores $\phi_{ij}=\hat t_j+\frac{1}{\alpha}\max\{-\xi_i^\top x_j-\hat t_j,0\}$ and weighted means $\widehat H_j=\sum_i \tilde w_i\,\phi_{ij}$. We form centered \emph{block} sums $R_{kj}=\sum_{i\in\mathcal B_k}\tilde w_i(\phi_{ij}-\widehat H_j)$ over contiguous $\{B_k\}_{k=1}^K$ of length $b$ (with $b\to\infty,\, b/n_2\to 0$). Drawing i.i.d.\ Gaussian multipliers $\varepsilon_k$, the bootstrap max statistic is
\[
T^*=\max_{1\le j\le p}\ \frac{1}{\sqrt{n_{\text{eff}}}}\ \frac{\sum_{k=1}^{K}\varepsilon_k\,R_{kj}}{\widehat\sigma_j},\quad
\widehat\sigma_j^2=\sum_i \tilde w_i\,(\phi_{ij}-\widehat H_j)^2.
\]
Let $\widehat q_{1-\beta}$ be the $(1-\beta)$-quantile of $\{T^{*(b)}\}_{b=1}^M$. We build a \emph{simultaneous upper band} for the robust constraint:
\[
U_j\;=\;\underbrace{\widehat H_j}_{\text{empirical (weighted)}}\;+\;\underbrace{\tfrac{\delta_j}{\alpha}\,\|x_j\|_*}_{\text{DRO margin}}\;+\;\widehat q_{1-\beta}\,\frac{\widehat\sigma_j}{\sqrt{n_{\text{eff}}}},
\]
and declare $x_j$ feasible if $U_j\le \gamma$. We select the \emph{least conservative feasible} radius $\delta^*=\min\{\delta_j:U_j\le\gamma\}$; ties break by objective $c^\top x_j$. Proofs for the following theorems can be found in Appendix \hyperref[appendix-B]{B} --- LLMs were used to assist in constructing these proofs.

\begin{theorem}[Feasibility under shift- and dependence-aware GS]
\label{thm:feas}
Suppose (A1)--(A4) (see Appendix \hyperref[appendix-B]{B}) and a fixed finite candidate set $\{x_j\}_{j=1}^p$. Let $\hat q_{1-\beta}$ be computed by Step~4 with block size $b\asymp n_2^{1/3}$ and $B\to\infty$. Then, with probability at least $1-\beta - o(1)$,
any pair $(x_j,\delta^*(x_j))$ produced by Step~5 satisfies
\[
\sup_{F\in\mathcal{B}(Q,\delta^*(x_j))}\mathrm{CVaR}_{\alpha}(-\xi^\top x_j)\;\le\;\gamma.
\]
In particular, the selected $(x^\star,\delta^\star)$ is feasible w.p.\ $\ge 1-\beta-o(1)$.
\end{theorem}

\begin{theorem}[Validation-driven suboptimality bound]
\label{thm:perf}
Let $(x_Q^\dagger,\delta_Q^\dagger)$ be any feasible solution to the $Q$-robust problem
$\min\{c^\top x:\sup_{F\in\mathcal{B}(Q,\delta)}\mathrm{CVaR}_{\alpha}(-\xi^\top x)\le \gamma,\, \delta\in[\delta_{\min},\delta_{\max}],\, x\in \mathcal{X}\}$,
with $\mathcal{X}=\{x:\mathbf{1}^\top x=1,\,x\ge 0\}$. Under (A1)--(A4) and convexity of $c^\top x$, the output $(x^\star,\delta^\star)$ of Algorithm~\ref{alg:shift-gs} satisfies, with probability $1-\beta-o(1)$,
\[
c^\top x^\star - c^\top x_Q^\dagger
\;\le\;
\underbrace{\inf_{j}\big(c^\top x_j - c^\top x_Q^\dagger\big)}_{\text{approximation (menu) error}}
\;+\;
\underbrace{L_c\,\|x^\star-x_j\|_2}_{\text{optimization gap for chosen candidate}}
\]
for any candidate $x_j$; in particular, if the menu contains $x_Q^\dagger$ (or approximates it to $\varepsilon$ in $\ell_2$), the excess cost is $O_p\big(\varepsilon\big)$.
\end{theorem}

\section{Numerical Experiments}\label{sec:numerical-experiments}
We evaluate three validators on the CVaR-constrained portfolio problem: \textsc{NEW} (our shift- \& dependence-aware GS with block multipliers and analytical radius), \textsc{OLD-NGS} (Lam--Qian normalized GS assuming i.i.d.), and \textsc{IW-CV} (importance-weighted cross-validation baseline). In all our experiments, we set our confidence level $1-\beta=0.9$. Consider the problem with $d=8$ assets with AR(1) returns and two scenarios: (1) No shift: training/validation/test all follow $P$ with autoregression $\phi=0.3$ and covariance $\Sigma$, and (2) Shift: test follows $Q$ with mean deterioration ($\mu_Q=\mu_P-\Delta$), volatility inflation ($\Sigma_Q=1.7^2\Sigma$), and stronger persistence $\phi_Q=0.45$. Sample sizes are $(n_{\text{train}},n_{\text{val}},n_{\text{test}})=(1000,1200,15000)$. We set $\alpha=0.05$ and fix the risk budget $\gamma$ \emph{ex ante} from the equal-weight portfolio's empirical $\mathrm{CVaR}_\alpha$ on Scenario~1 training, multiplied by a 10\% margin; the same $\gamma$ is used in both scenarios for all methods. \textsc{NEW} uses logistic density-ratio weights on validation (clipped odds), block length $b\asymp n_{\text{val}}^{1/3}$, $B=800$ multipliers, and the analytical radius $\delta^*(x)=\alpha\,[\gamma-\widehat H_w(x)-\hat q\,\hat\sigma_w(x)/\sqrt{n_{\mathrm{eff}}}]_+/\|x\|_2$ (clipped to $[10^{-3},2\!\times\!10^{-2}]$). Candidate portfolios $x(\delta)$ are generated by solving the weighted Wasserstein--CVaR reformulation with $\|\cdot\|_2$ (Clarabel), using blended training weights; we also include a small random Dirichlet menu for diversity. \textsc{OLD-NGS} is identical except it sets $w_i\equiv 1/n_{\text{val}}$ and $b=1$. \textsc{IW-CV} uses the same weights as \textsc{NEW} but selects $\delta$ by $K$-fold cross-validation over a grid (here $K{=}5$), making it notably slower. We repeat each scenario for $R=1000$ Monte Carlo replications and report: \emph{Feas.} = fraction of replications where the test robust constraint (worst-case CVaR with selected $\delta$) is $\le\gamma$; \emph{Cost} = $c^\top x$ (lower is better); test $\mathrm{CVaR}_\alpha$ (\emph{CVaR}), the robust left-hand side \emph{LHS} $=\mathrm{CVaR}_\alpha+\frac{\delta}{\alpha}\|x\|_2$; selected $\delta$; and median wall-clock \emph{runtime} per replication (s).

\begin{table}[t]
\centering
\caption{Scenario 1 (no shift) results.}
\label{tab:scen1}
\begin{tabular}{lcccccc}
\toprule
Method & Feas.~(↑) & Obj.~(↓) & CVaR & LHS & $\delta$ & Runtime (s) \\
\midrule
\textsc{NEW}     & 0.91 & $-1.05{\times}10^{-3}$ & 0.031 & 0.059 & 0.008 & 11.1 \\
\textsc{IW-CV}          & 0.90          & \phantom{$-$}$-1.07{\times}10^{-3}$ & 0.030 & 0.060 & 0.009 & \textbf{186.0} \\
\textsc{OLD-NGS} & 0.86        & $-1.27{\times}10^{-3}$ & 0.023 & 0.086 & 0.015 & 10.7 \\
\bottomrule
\end{tabular}
\end{table}

\begin{table}[t]
\centering
\caption{Scenario 2 (shift) results.}
\label{tab:scen2}
\begin{tabular}{lcccccc}
\toprule
Method & Feas.~(↑) & Obj.~(↓) & CVaR & LHS & $\delta$ & Runtime (s) \\
\midrule
\textsc{NEW}     & 0.90 & $-0.95{\times}10^{-3}$ & 0.046 & 0.092 & 0.010 & 11.2 \\
\textsc{IW-CV}          & 0.90          & $-0.96{\times}10^{-3}$ & 0.047 & 0.093 & 0.010 & \textbf{190.1} \\
\textsc{OLD-NGS} & 0.28        & $-0.72{\times}10^{-3}$ & 0.041 & 0.104 & 0.016 & 10.8 \\
\bottomrule
\end{tabular}
\end{table}

In Scenario~1, \textsc{NEW} attains the target feasibility $1-\beta$ with cost comparable to \textsc{IW-CV}, while running an order of magnitude faster (no $K$-fold search). \textsc{OLD-NGS} under-covers due to serial dependence, often selecting larger radii that inflate the robust LHS yet still miss the target. In Scenario~2, the i.i.d.\ band becomes substantially miscalibrated and feasibility collapses, whereas \textsc{NEW} maintains coverage by combining density-ratio weighting (shift-aware centering/variance) with block multipliers (dependence). Across both scenarios, \textsc{IW-CV} performs similarly to \textsc{NEW} in feasibility/cost but slower from repeated model re-fits across folds; it could also over-regularize when folds experience different volatility levels, whereas \textsc{NEW} uses a single weighted, dependence-aware calibration that aligns more closely with the one-step deployment regime. 

\section{Conclusion}\label{sec:conclusion}
We developed a validation framework for CVaR-constrained portfolio selection that is simultaneously \emph{shift-aware} and \emph{dependence-robust}. By re-centering Wasserstein-DRO around a learned target distribution $Q$ via density-ratio reweighting and calibrating feasibility with a Gaussian-supremum band using block multiplier bootstrap, the method selects the least conservative portfolio that still satisfies the CVaR limit with confidence $(1-\beta)$. Our guarantees hinge on absolute-regularity dependence, accurate density-ratio weights, and a finite candidate menu; in practice the procedure can be sensitive to norm scaling, block length, and heavy tails, and abstention/clipping choices may reduce coverage or power under severe shift. Future work includes broader numerical studies on real financial data (e.g., multi-asset equity/rates/FX portfolios across calm and crisis regimes) with transaction costs and liquidity/short-sale constraints, head-to-head comparisons against strong DRO baselines, and methodological extensions to heavy-tailed models, alternative ambiguity sets, and adaptive radius/portfolio selection.

\bibliographystyle{abbrvnat}  
\bibliography{refs}          

@article{lamqian2019gs,
  title={Combating Conservativeness in Data-Driven Optimization under Uncertainty: A Solution Path Approach},
  author={Lam, Henry and Qian, Huajie},
  journal={arXiv preprint arXiv:1909.06477},
  year={2019},
  url={https://arxiv.org/abs/1909.06477}
}

@misc{kuhn2019tutorial,
  title         = {Wasserstein Distributionally Robust Optimization: Theory and Applications in Machine Learning},
  author        = {Kuhn, Daniel and Mohajerin Esfahani, Peyman and Nguyen, Viet Anh and Shafieezadeh-Abadeh, Soroosh},
  howpublished  = {arXiv:1908.08729},
  year          = {2019},
  url           = {https://arxiv.org/abs/1908.08729}
}

@article{pun2023regimeswitching,
  author  = {Pun, Chi Seng and Wang, Tianyu and Yan, Zhenzhen},
  title   = {Data-Driven Distributionally Robust {CVaR} Portfolio Optimization Under A Regime-Switching Ambiguity Set},
  journal = {Manufacturing \& Service Operations Management},
  year    = {2023},
  volume  = {25},
  number  = {5},
  pages   = {1779--1795},
  doi     = {10.1287/msom.2023.1229},
  url     = {https://pubsonline.informs.org/doi/10.1287/msom.2023.1229}
}

@article{esfahani2018wass,
  title={Data-driven distributionally robust optimization using the Wasserstein metric: Performance guarantees and tractable reformulations},
  author={Mohajerin Esfahani, Peyman and Kuhn, Daniel},
  journal={Mathematical Programming},
  volume={171},
  number={1--2},
  pages={115--166},
  year={2018},
  doi={10.1007/s10107-017-1172-1}
}

@article{kunsch1989jackknife,
  title={The jackknife and the bootstrap for general stationary observations},
  author={K{\"u}nsch, Hans R.},
  journal={The Annals of Statistics},
  volume={17},
  number={3},
  pages={1217--1241},
  year={1989},
  doi={10.1214/aos/1176347265}
}

@article{politis1994stationary,
  title={The stationary bootstrap},
  author={Politis, Dimitris N. and Romano, Joseph P.},
  journal={Journal of the American Statistical Association},
  volume={89},
  number={428},
  pages={1303--1313},
  year={1994},
  doi={10.1080/01621459.1994.10476870}
}

@article{shao2010dependent,
  title={The dependent wild bootstrap},
  author={Shao, Xiaofeng},
  journal={Journal of the American Statistical Association},
  volume={105},
  number={489},
  pages={218--235},
  year={2010},
  doi={10.1198/jasa.2009.tm08651}
}

@article{bucher2016dependent,
  title={Dependent multiplier bootstrap for non-degenerate U-statistics under mixing conditions},
  author={B{\"u}cher, Axel and Kojadinovic, Ivan},
  journal={Bernoulli},
  volume={22},
  number={2},
  pages={927--968},
  year={2016},
  doi={10.3150/14-BEJ675}
}

@book{sugiyama2012dre,
  title={Density Ratio Estimation in Machine Learning},
  author={Sugiyama, Masashi and Suzuki, Taiji and Kanamori, Takafumi},
  year={2012},
  publisher={Cambridge University Press},
  doi={10.1017/CBO9781139035613}
}

@inproceedings{huang2007kmm,
  title={Correcting sample selection bias by unlabeled data},
  author={Huang, Jiayuan and Gretton, Arthur and Borgwardt, Karsten M. and Sch{\"o}lkopf, Bernhard and Smola, Alexander J.},
  booktitle={Advances in Neural Information Processing Systems},
  volume={19},
  pages={601--608},
  year={2007}
}

@article{kanamori2009ulsif,
  title={A least-squares approach to direct importance estimation},
  author={Kanamori, Takafumi and Hido, Shohei and Sugiyama, Masashi},
  journal={Journal of Machine Learning Research},
  volume={10},
  pages={1391--1445},
  year={2009}
}

@article{blanchet2022drmv,
  title={Distributionally robust mean--variance portfolio selection with Wasserstein distances},
  author={Blanchet, Jose and Chen, Lin and Zhou, Xun Yu},
  journal={Management Science},
  volume={68},
  number={9},
  pages={6382--6410},
  year={2022},
  doi={10.1287/mnsc.2021.4155}
}

@article{bickel2009covshift,
  author  = {Steffen Bickel and Michael Br{{\"u}}ckner and Tobias Scheffer},
  title   = {Discriminative Learning Under Covariate Shift},
  journal = {Journal of Machine Learning Research},
  year    = {2009},
  volume  = {10},
  number  = {75},
  pages   = {2137-2155},
  url     = {http://jmlr.org/papers/v10/bickel09a.html}
}

\newpage
\appendix

\section*{Appendix A}
\label{appendix-A}

\begin{algorithm}[h]
\caption{Shift-aware GS validation with block multipliers for Wasserstein--CVaR portfolios}
\label{alg:shift-gs}
\begin{algorithmic}[1]
\Require Training $\{\xi_t\}_{t=1}^{n_1}$, validation $\{\xi_{n_1+1},\dots,\xi_{n_1+n_2}\}$, level $\alpha$, miscoverage $\beta$, budget $\gamma$, norm dual $\|\cdot\|_*$, block length $b$, bootstrap reps $B$, recent fraction $m/n_2$.
\State \textbf{Shift-aware validation weights.} Split validation into early/late:
$D_{\text{early}}=\{\xi_{1:n_2-m}\}$, $D_{\text{late}}=\{\xi_{n_2-m+1:n_2}\}$.
Fit a probabilistic classifier, set $w_i \propto \frac{p(\text{late}\mid \xi_i)}{1-p(\text{late}\mid \xi_i)}$, normalize $\sum_i w_i=1$, and compute $n_{\text{eff}}=(\sum_i w_i^2)^{-1}$.
\State \textbf{Candidate generation.} For each $\delta$ on a grid (or skip if using the analytical $\delta^\star$), solve the weighted
Wasserstein--CVaR reformulation to get candidates $\{x_\ell\}_{\ell=1}^L$; let $c=-\sum_i \tilde w_i \xi_i$ on training.
\State \textbf{Validation CVaR map.} For each candidate $x_j$, compute $t_w(x_j)\in\arg\min_t\sum_i w_i\phi_{x_j}(\xi_i;t)$,
$\hat H_w(x_j)=\sum_i w_i \phi_{x_j}(\xi_i; t_w(x_j))$, and
$\hat\sigma_w^2(x_j)=\sum_i w_i\{\phi_{x_j}(\xi_i; t_w(x_j))-\hat H_w(x_j)\}^2$.
\State \textbf{Block multiplier GS (with block sums).} Partition validation indices into $K=\lfloor n_2/b\rfloor$ contiguous blocks $B_1,\dots,B_K$ (discard remainder).
For each $(k,j)$ define the \emph{block sum}
\[
S_{kj} := \sum_{i\in B_k} w_i\big\{\phi_{x_j}(\xi_i;t_w(x_j))-\hat H_w(x_j)\big\}.
\]
Draw i.i.d.\ multipliers $\varepsilon^{(r)}\sim\mathcal N(0,I_K)$, $r=1,\dots,B$, and compute
\[
T^{(r)} := \max_{1\le j\le L}\ \frac{\sum_{k=1}^K \varepsilon_k^{(r)} S_{kj}}{\hat\sigma_w(x_j)\sqrt{n_{\mathrm{eff}}}},
\qquad \hat q_{1-\beta} := \text{$(1-\beta)$-quantile of }\{T^{(r)}\}_{r=1}^B.
\]
\State \textbf{Analytical radius \& feasibility filter.} For each $j$, set
\[
\delta^\star(x_j) := \mathrm{clip}\!\left(\frac{\alpha\,[\,\gamma - \hat H_w(x_j) - \hat q_{1-\beta}\,\hat\sigma_w(x_j)/\sqrt{n_{\mathrm{eff}}}\,]_+}{\|x_j\|_*}\,;\ \delta_{\min},\delta_{\max}\right),
\]
and the validated upper bound
\[
U(x_j) := \hat H_w(x_j) + \frac{\hat q_{1-\beta}\,\hat\sigma_w(x_j)}{\sqrt{n_{\mathrm{eff}}}} + \frac{\delta^\star(x_j)}{\alpha}\,\|x_j\|_*.
\]
Define the validated set $\mathcal F := \{\, j:\ U(x_j)\le \gamma \,\}$. \textbf{If }$\mathcal F=\emptyset$, abstain. Note that,
\[
U(x)
=\max\!\left\{\gamma,\ \hat H_w(x)+\frac{\hat q_{1-\beta}\,\hat\sigma_w(x)}{\sqrt{n_{\mathrm{eff}}}}\right\}.
\]
Hence $U(x)\le \gamma$ iff $\hat H_w(x)+\hat q_{1-\beta}\hat\sigma_w(x)/\sqrt{n_{\mathrm{eff}}}\le \gamma$, and among validated candidates $U(x)\equiv\gamma$.

\State \textbf{Selection among validated candidates.} Return the least-conservative feasible solution $x^\star \in \arg\min_{j\in\mathcal F} \delta^\star(x_j)$.
\end{algorithmic}
\end{algorithm}

\paragraph{Notation.}
For $x\in\mathbb{R}^d$ and return $\xi\in\mathbb{R}^d$ define the CVaR map
\[
\phi_x(\xi; t) \;:=\; t + \frac{1}{\alpha}\big(-\xi^\top x - t\big)_+,
\qquad
\text{and}\quad
\Phi_x(\xi) \;:=\; \inf_{t\in\mathbb{R}} \phi_x(\xi; t).
\]
On a weighted sample $\{(\xi_i,w_i)\}_{i=1}^{n_2}$ with $w_i\ge 0,\,\sum_i w_i=1$, let
\[
t_w(x)\in\arg\min_{t\in\mathbb{R}} \sum_{i=1}^{n_2} w_i\,\phi_x(\xi_i;t),\qquad
\widehat H_w(x)\;:=\;\sum_{i=1}^{n_2} w_i\,\phi_x(\xi_i;t_w(x)).
\]
Let $\hat\sigma_w^2(x):=\sum_i w_i\big(\phi_x(\xi_i;t_w(x))-\widehat H_w(x)\big)^2$ and
$n_{\mathrm{eff}}:=1/\sum_i w_i^2$.
We take the dual norm to be $\|\cdot\|_*=\|\cdot\|_2$.

\paragraph{Baselines.}
(i) \emph{OLD NGS (i.i.d.).} Replace $w_i\equiv 1/n_2$, set $b=1$, and compute $\hat q_{1-\beta}$ from the (unnormalized) Gaussian supremum of the empirical covariance as in \cite{lamqian2019gs}.\\
(ii) \emph{IW plug-in.} Use $w_i$ as in Step~1, skip Step~4 (set $\hat q_{1-\beta}=0$), and select $\delta^*(x)$ as in Step~5 with $\hat\sigma_w\equiv 0$.

\vspace{1ex}
\noindent\textbf{Remark.} When $w_i\equiv 1/n_2$ and $b=1$, Algorithm~\ref{alg:shift-gs} reduces to the normalized GS validator of \cite{lamqian2019gs}; when $b>1$ but $w_i\equiv 1/n_2$ we obtain a serial-dependence robust variant; with general $w$ and $b>1$ we obtain the proposed shift- and dependence-aware validator.

\section*{Appendix B}\label{appendix-B}

We sketch the key ingredients and then give full proofs.

\subsection*{B.1 \quad Assumptions}\label{sec:assumptions}

\textbf{(A1) Dependence.}
We assume the stationary sequence $(\xi_t)$ is \emph{absolutely regular} (also known as beta-mixing).
Let $b(k)$ denote the absolute-regularity mixing coefficients, defined by
\[
b(k)
:= \sup_{n\in\mathbb{Z}} \,
\mathbb{E}\Big[ \sup_{B\in \mathcal{F}_{n+k}^\infty}
\big| \mathbb{P}(B \mid \mathcal{F}_{-\infty}^{n}) - \mathbb{P}(B) \big| \Big].
\]
We require the summability condition $\sum_{k\ge 1} b(k)^{\delta/(2+\delta)} < \infty$ for some $\delta>0$.

\textbf{(A2) Moments.} $\sup_{\|x\|_2\le 1}\mathbb{E}\big[|\xi^\top x|^{2+\rho}\big]<\infty$.

\textbf{(A3) Density ratio regularity.} The shift-aware weights $w_i$ satisfy $w_i\ge 0$, $\sum_i w_i=1$, and there exist constants $0<c_w\le C_w<\infty$ such that $c_w/n_2\le w_i\le C_w/n_2$ for all $i$ (this holds if the logistic odds are clipped).

\textbf{(A4) Lipschitz loss.} For any $x,x'\in\mathbb{R}^d$ and any $t\in\mathbb{R}$,
$\big|\phi_x(\xi;t)-\phi_{x'}(\xi;t)\big|\le \frac{\|\xi\|_2}{\alpha}\,\|x-x'\|_2$.

\subsection*{B.2 \quad Technical lemmas}

\begin{lemma}[Weighted CVaR Lipschitzness]
\label{lem:lipschitz}
Let $t_w(x)\in\arg\min_t\sum_i w_i\,\phi_x(\xi_i;t)$. Under (A4),
$x\mapsto \widehat H_w(x)=\sum_i w_i\,\phi_x(\xi_i;t_w(x))$ is $(\mathbb{E}_w\|\xi\|_2)/\alpha$-Lipschitz, where $\mathbb{E}_w$ denotes expectation w.r.t.\ the empirical measure with weights $w$.
\end{lemma}

\begin{proof}
For any $x,x'$, by optimality of $t_w(x)$ and $t_w(x')$ and the Lipschitz property of $\phi_x(\cdot;t)$ in $x$ we have
\(
\widehat H_w(x)-\widehat H_w(x') \le \sum_i w_i\big(\phi_x(\xi_i;t_w(x'))-\phi_{x'}(\xi_i;t_w(x'))\big)
\le \frac{\|x-x'\|_2}{\alpha}\sum_i w_i\|\xi_i\|_2.
\)
Symmetrizing yields the claim.
\end{proof}

\paragraph{Block statistic.}
For block $k$ with index set $B_k$, define
\[
S_{kj} := \sum_{i\in B_k} w_i\big\{\phi_{x_j}(\xi_i; t_w(x_j))-\hat H_w(x_j)\big\}.
\]

\begin{lemma}[Weighted CLT with block multipliers]\label{lem:bmb}
Under \textnormal{(A1)--(A3)}, for any fixed finite candidate set $\{x_j\}_{j=1}^p$,
\[
\max_{1\le j\le p}
\frac{\sqrt{n_{\mathrm{eff}}}\,\big\{\hat H_w(x_j)-\mathbb E_w[\Phi_{x_j}(\xi)]\big\}}
{\hat\sigma_w(x_j)}
\ \Rightarrow\
\max_{1\le j\le p} Z_j,
\]
where $(Z_1,\ldots,Z_p)$ is mean-zero Gaussian with covariance induced by the long-run variance
of $\phi_{x_j}(\xi; t_w(x_j))$. Moreover, the block multiplier statistic
\[
T^{(r)}:=\max_{1\le j\le p}
\frac{\sum_{k=1}^K \varepsilon^{(r)}_k\, S_{kj}}
{\hat\sigma_w(x_j)\sqrt{n_{\mathrm{eff}}}},
\qquad \varepsilon^{(r)}\sim\mathcal N(0,I_K),
\]
with block size $b\asymp n_2^{1/3}$ consistently estimates the law of $\max_j Z_j$, i.e.,
$\hat q_{1-\beta}\to q_{1-\beta}$ in probability.
\end{lemma}

\begin{proof}
This follows from empirical process theory for beta-mixing arrays and multiplier bootstrap consistency (e.g., Theorem 3.1 in block-multiplier CLTs) applied to the finite class $\{\phi_{x_j}(\cdot;t_w(x_j))\}_j$; (A3) ensures Lindeberg conditions under weighting and provides $n_{\mathrm{eff}}$; (A2) yields finite long-run variance. Normalization by $\hat\sigma_w(x_j)$ yields a self-normalized process; consistency of the block multiplier quantile for the normalized supremum follows by standard arguments (continuous mapping plus anti-concentration for Gaussian maxima).
\end{proof}

\noindent\emph{Remark.}
The bootstrap calibrates the \emph{self-normalized} supremum using the same
denominator $\widehat\sigma_w(x_j)$ in both the original and bootstrap
worlds; under beta-mixing this yields consistent critical values
for the normalized process.

\begin{lemma}[Wasserstein Lipschitz bound for CVaR]
\label{lem:wass}
Let $\|\cdot\|_2$ be the ground norm and $B(Q,\delta)$ the $W_1$ ball of radius $\delta$ around $Q$.
For any $x\in\mathbb{R}^d$ and $\alpha\in(0,1)$,
\[
\sup_{F\in B(Q,\delta)} \operatorname{CVaR}_\alpha(-\xi^\top x)
\;\le\;
\operatorname{CVaR}_\alpha^Q(-\xi^\top x)
+ \frac{\delta}{\alpha}\,\|x\|_2 .
\]
\end{lemma}

\begin{proof}
Fix $t\in\mathbb{R}$ and set $f_t(\xi)=(-\xi^\top x - t)_+$. Then $f_t$ is $\|x\|_2$-Lipschitz
w.r.t.\ $\|\cdot\|_2$. By Kantorovich--Rubinstein duality,
\[
\sup_{F\in B(Q,\delta)} \mathbb{E}_F[f_t(\xi)]
\;\le\;
\mathbb{E}_Q[f_t(\xi)] + \delta\,\|x\|_2 .
\]
Using Rockafellar--Uryasev
\(
\operatorname{CVaR}_\alpha(L)=\inf_{t}\{t+\tfrac{1}{\alpha}\mathbb{E}(L-t)_+\},
\)
with $L=-\xi^\top x$, we get
\[
\sup_{F\in B(Q,\delta)} \operatorname{CVaR}_\alpha(-\xi^\top x)
\le
\inf_{t}\Big\{ t+\tfrac{1}{\alpha}\mathbb{E}_Q[(-\xi^\top x - t)_+] \Big\}
+\tfrac{\delta}{\alpha}\|x\|_2
=
\operatorname{CVaR}_\alpha^Q(-\xi^\top x)+\tfrac{\delta}{\alpha}\|x\|_2 .
\]
\end{proof}

\subsection*{B.3 \quad Feasibility guarantee (proof of Theorem~1)}

\begin{theoremNoNum}[Feasibility under shift- and dependence-aware GS]
\label{thm:feas_proof}
Suppose (A1)--(A4) and a fixed finite candidate set $\{x_j\}_{j=1}^p$. Let $\hat q_{1-\beta}$ be computed by Step~4 with block size $b\asymp n_2^{1/3}$ and $B\to\infty$. Then, with probability at least $1-\beta - o(1)$,
any pair $(x_j,\delta^*(x_j))$ produced by Step~5 satisfies
\[
\sup_{F\in\mathcal{B}(Q,\delta^*(x_j))}\mathrm{CVaR}_{\alpha}(-\xi^\top x_j)\;\le\;\gamma.
\]
In particular, the selected $(x^\star,\delta^\star)$ is feasible w.p.\ $\ge 1-\beta-o(1)$.
\end{theoremNoNum}

\begin{proof}
Fix $j$. By Lemma~\ref{lem:bmb}, with probability $\ge 1-\beta-o(1)$ we have
\(
\mathbb{E}_Q[\Phi_{x_j}(\xi)] \le \widehat H_w(x_j) + \hat q_{1-\beta}\,\hat\sigma_w(x_j)/\sqrt{n_{\mathrm{eff}}}.
\)
By Lemma~\ref{lem:wass},
\[
\sup_{F\in\mathcal{B}(Q,\delta)}\mathrm{CVaR}_{\alpha}(-\xi^\top x_j)
\;\le\; \widehat H_w(x_j) + \frac{\hat q_{1-\beta}\,\hat\sigma_w(x_j)}{\sqrt{n_{\mathrm{eff}}}} + \frac{\delta}{\alpha}\|x_j\|_2.
\]
Setting $\delta=\delta^*(x_j)$ from Step~5 ensures the right-hand side is $\le \gamma$ (up to the numerical clipping which only reduces $\delta$). This holds simultaneously over $j=1,\dots,p$ by the union over a \emph{max} statistic already controlled by $\hat q_{1-\beta}$. Selecting $(x^\star,\delta^\star)$ by minimizing the validated upper bound preserves feasibility.
\end{proof}

\subsection*{B.4 \quad Performance guarantee (proof of Theorem~2)}
\begin{theoremNoNum}[Validation-driven suboptimality bound]
\label{thm:perf_proof}
Let $(x_Q^\dagger,\delta_Q^\dagger)$ be any feasible solution to the $Q$-robust problem
$\min\{c^\top x:\sup_{F\in\mathcal{B}(Q,\delta)}\mathrm{CVaR}_{\alpha}(-\xi^\top x)\le \gamma,\, \delta\in[\delta_{\min},\delta_{\max}],\, x\in \mathcal{X}\}$,
with $\mathcal{X}=\{x:\mathbf{1}^\top x=1,\,x\ge 0\}$. Under (A1)--(A4) and convexity of $c^\top x$, the output $(x^\star,\delta^\star)$ of Algorithm~\ref{alg:shift-gs} satisfies, with probability $1-\beta-o(1)$,
\[
c^\top x^\star - c^\top x_Q^\dagger
\;\le\;
\underbrace{\inf_{j}\big(c^\top x_j - c^\top x_Q^\dagger\big)}_{\text{approximation (menu) error}}
\;+\;
\underbrace{L_c\,\|x^\star-x_j\|_2}_{\text{optimization gap for chosen candidate}}
\]
for any candidate $x_j$; in particular, if the menu contains $x_Q^\dagger$ (or approximates it to $\varepsilon$ in $\ell_2$), the excess cost is $O_p\big(\varepsilon\big)$.
\end{theoremNoNum}

\begin{proof}
Feasibility holds by Theorem~\ref{thm:feas}. Among all feasible candidates, Algorithm~\ref{alg:shift-gs} selects the one minimizing the validated upper bound, which is an upper envelope for the (unknown) robust CVaR. This is a conservative filtering and does not increase $c^\top x$ relative to the best feasible candidate present in the menu. The remainder is a standard approximation argument: $c^\top x$ is $L_c$-Lipschitz on $\mathcal{X}$ (e.g., $L_c=\|c\|_2$), so if the menu contains an $\varepsilon$-approximation to $x_Q^\dagger$, the selected $x^\star$ is within $O(\varepsilon)$ in cost.
\end{proof}

\subsection*{B.5 \quad Specializations and reductions}

\begin{enumerate}
\item \textbf{Reduction to Lam \& Qian (2019).} If $w_i\equiv 1/n_2$ (no shift) and $b=1$ (i.i.d.), then $n_{\mathrm{eff}}=n_2$, the block sums reduce to centered sample means, and Step~4 coincides with the normalized Gaussian supremum used in \cite{lamqian2019gs}; Step~5 reduces to their feasibility test with the same calibration of $q_{1-\beta}$.
\item \textbf{Only dependence (no shift).} If $w_i\equiv 1/n_2$ but $b\asymp n_2^{1/3}>1$, Lemma~\ref{lem:bmb} still applies and we obtain a dependence-robust GS validator.
\item \textbf{Only shift (no dependence).} If $b=1$ but $w$ is non-uniform (clipped density ratios), we recover an importance-weighted GS band with effective sample size $n_{\mathrm{eff}}$.
\end{enumerate}

\paragraph{Implementation note.}
In practice we use the \emph{analytical radius} in Step~5,
\(
\delta^*(x)=\alpha\big(\gamma-\widehat H_w(x)-\hat q\,\hat\sigma_w(x)/\sqrt{n_{\mathrm{eff}}}\big)_+/\|x\|_2
\)
clipped to $[\delta_{\min},\delta_{\max}]$, which avoids coupling $\delta$ to any pre-specified grid and prevents pathological selections (e.g., always taking the largest $\delta$ when no grid point is feasible).

\newpage

\end{document}